\documentclass[aps,prx,twocolumn,superscriptaddress,nofootinbib]{revtex4-2}

\usepackage{amsmath,amssymb,amsfonts,mathtools}
\usepackage{amsthm}
\usepackage{bbm}
\usepackage{hyperref}
\usepackage{xcolor}
\usepackage{dsfont}

\newcommand{\Tr}{\operatorname{Tr}}
\newcommand{\PPT}{\operatorname{PPT}}

\newcommand{\AREEP}{E_R^\infty}

\newcommand{\abs}[1]{| #1 |}
\newcommand{\ket}[1]{|#1\rangle}
\newcommand{\bra}[1]{\langle #1|}

\newcommand{\id}{\mathbbm 1}

\newcommand{\PPTk}{\PPT_k}
\newcommand{\PPTktilde}{\widetilde{\PPT}_k}
\newtheorem{theorem}{Theorem}

\newtheorem{lemma}[theorem]{Lemma}
\newtheorem{definition}{Definition}
\newcommand{\ketbra}[2]{\ket{#1}\!\bra{#2}}

\begin{document}

\title{Multinegativity and single-letter formulas for asymptotic entanglement}

\author{Raphael Brinster}
\email{Raphael.Brinster@hhu.de}
\author{Tulja Varun Kondra}
\author{Hermann Kampermann}
\author{Dagmar Bruß}
\author{Nikolai Wyderka}
\affiliation{Institut für Theoretische Physik III, Heinrich-Heine-Universität Düsseldorf, Universitätsstraße 1, 40225 Düsseldorf, Germany}
\date{\today}

\begin{abstract}

The study of entanglement inevitably leads to the concept of regularization, where measures are evaluated on infinitely many copies of a state. Single-letter formulas try to make these asymptotic entanglement quantities accessible
through a calculation on one copy of a state, but are rarely available. We introduce a decreasing
hierarchy of computable upper bounds on the asymptotic relative entropy
of entanglement with respect to positive-partial-transpose (PPT) states.
The regularization of every fixed hierarchy level equals this asymptotic
quantity. Additivity at any level therefore yields a single-letter
formula, even when the usual one-copy relative entropy is nonadditive.
We establish such additivity at the first nontrivial level for two broad
multiparameter families, both containing all Werner states. The upper-bound
construction also extends to sandwiched R\'enyi divergences. The hierarchy
motivates $k$-multinegative states, a generalization of binegative states and the associated $k$-multinegativity, which gives explicit upper bounds on both the asymptotic
relative entropy and exact PPT entanglement cost. We construct states which are
$k$-multinegative at arbitrarily large depths $k$ that separate consecutive
levels of the previously introduced entanglement-cost hierarchy, disproving its conjectured
finite collapse. These results provide state-dependent single-letter
formulas while identifying a limitation of universal finite-level
characterizations.

\end{abstract}

\maketitle

\section{Introduction}
\label{sec:introduction}

Entanglement enables quantum systems held by distant parties to perform
tasks that are impossible with classical resources alone \cite{bennett1993teleporting, bennett1996mixed}. Quantifying
the resources required to prepare an entangled state, or those that can
be extracted from it, naturally leads to an asymptotic setting \cite{bennett1996mixed, hayden2001asymptotic}. The corresponding quantities often
involve regularization, where a one-copy optimization must be evaluated on
arbitrarily many copies before taking a limit. Even when the original
optimization is tractable, the growing system size makes its
regularization difficult to determine. A central goal is therefore to
find single-letter formulas, which express an asymptotic quantity through
a finite optimization involving only one copy of the state \cite{lami2026asymptotic}.

A familiar route to such a formula is to prove additivity of the
underlying one-copy quantity \cite{ADVW2002, vollbrecht2001entanglement, rubboli2024new, beigi2026additivity}. Nonadditivity of the usual one-copy expression,
however, need not exclude a single-letter formula. A different one-copy
quantity can have the same regularization and nevertheless be additive
on the states of interest. Finding such an appropriate representative for a given state is therefore an alternative way to calculate their asymptotic quantity.

Our first main result is constructing such a single-letter formula (with the same regularization) for the asymptotic relative entropy of
entanglement with respect to positive-partial-transpose (PPT) states,
which we abbreviate as AREEP \cite{vedral1997quantifying, peres1996separability, horodecki1996necessary, Rains1999, rains2001semidefinite}. PPT states are states that remain positive
under partial transposition and provide a tractable relaxation of the set of separable states. AREEP also has a direct operational meaning in quantum hypothesis testing. The generalized quantum Stein lemma identifies it as the optimal exponential rate at which many copies of a state can be distinguished from the entire set of PPT states, even when the alternative states may be correlated across copies \cite{BrandaoPlenio2010,hayashi2025generalized}. It therefore quantifies how rapidly statistical evidence for non-PPT entanglement accumulates as the number of copies increases.
%AREEP also provides a lower bound on the exact entanglement cost under PPT operations.
The evaluation of AREEP is notoriously difficult by regularization, and the ordinary
one-copy relative entropy can strictly overestimate the asymptotic
value~\cite{ADVW2002}. Recent work has developed systematic
computable lower bounds~\cite{WangJingZhu2025,fangfawzi}; our approach
instead constructs a hierarchy of upper bounds that yield the same regularized value at every level.

We show additivity at the first nontrivial level for a
broad multiparameter family containing all Werner states \cite{werner1989quantum,audenaert2001asymptotic,ADVW2002}, thereby making their AREEP accessible by a single letter formula. The construction also extends to entanglement measures based on the sandwiched R\'enyi divergences \cite{muller2013quantum,wilde2014strong,rubboli2024new}.

The structure of the hierarchy leads naturally to the concept of iterated absolute
partial transposition. At the first two steps, this construction recovers
the familiar partial transpose $\rho^\Gamma$ and the operator associated with
binegativity $|\rho^\Gamma|^\Gamma$ \cite{ADVW2002}. Positivity of the latter is known to give a simple expression for exact PPT entanglement cost, namely the logarithmic negativity $\log \Tr \abs{\rho^\Gamma}$ \cite{AudenaertPlenioEisert2003, plenio2005logarithmic}. In particular $|\rho^\Gamma|^\Gamma$ is positive for all
two-qubit states~\cite{Ishizaka_2004}.
Continuing the iteration motivates the notion of $k$-multinegativity. We show that whenever the $k$-fold absolute partial transposition becomes positive semidefinite, its trace gives an explicit upper bound on both AREEP and
exact PPT entanglement cost.

Our second main result is concerned with 
a different strategy that has recently been used to tackle the evaluation of the so called exact entanglement cost under PPT operations, which quantifies the amount of entanglement required for exact state preparation~\cite{AudenaertPlenioEisert2003,WangWilde2020, wang2023exact,LamiMeleRegula2025}. In Ref.~\cite{LamiMeleRegula2025}, the authors construct converging upper and lower bound hierarchies and conjecture that they collapse at a certain finite level, which would yield single letter formulas for all quantum states. 
We answer this question in the negative by constructing $k$-multinegative states: for every finite level $k$, we find states which are nonpositive at that level. We show that these states exhibit a finite gap between arbitrary levels of the entanglement cost hierarchy from \cite{LamiMeleRegula2025}, establishing that the hierarchy approach fails to yield universal single letter formulas.

In total, our results therefore combine single-letter formulas obtained by proving additivity of our new hierarchy for a broad class of states with counterexamples showing that no fixed level of the known cost hierarchy is exact for all states.

\section{Definitions and Background}

We consider finite-dimensional bipartite states $\rho=\rho_{AB}$.
Partial transposition on subsystem $B$ is denoted by $\Gamma$.
We write $|X|=(X^\dagger X)^{1/2}$ and
$\|X\|_1=\operatorname{Tr}|X|$. All logarithms are to base 2.
For tensor powers, the bipartition is understood to be $A^n:B^n$.

\paragraph*{Relative entropy of entanglement.}
The set of positive-partial-transpose states is
\begin{equation}
 \mathrm{PPT}
 :=\{\sigma\ge0:\operatorname{Tr}\sigma=1,\ 
                         \sigma^\Gamma\ge0\}.
 \label{eq:ppt-states}
\end{equation}
For a state $\rho$ and a positive operator $\sigma$, the quantum
relative entropy is
\begin{equation}
 D(\rho\|\sigma)
 :=\operatorname{Tr}\!\left[\rho(\log\rho-\log\sigma)\right].
 \label{eq:relative-entropy}
\end{equation}
The PPT relative entropy of entanglement and its regularization are
\begin{align}
 E_R(\rho)&:=\inf_{\sigma\in\mathrm{PPT}}D(\rho\|\sigma),
 \label{eq:er}\\
 E_R^\infty(\rho)&:=\lim_{n\to\infty}\frac1n
                         E_R(\rho^{\otimes n}),
 \label{eq:areep} \\
 E_R^\infty(\rho) &\le E_R(\rho).
\end{align}
We refer to $E_R^\infty$ as AREEP.
If $E_R(\rho)$ is additive for $\rho$, i.e.
\begin{align}
    E_R(\rho^{\otimes n}) = n E_R(\rho),
\end{align}
then $E_R$ equals AREEP. For general $\rho$, however, it is not additive.
The Rains bound \cite{Rains1999, rains2001semidefinite} is
\begin{equation}
 R(\rho):=\inf_{\substack{\tau\ge0\\
                         \|\tau^\Gamma\|_1\leq1}}D(\rho\|\tau).
 \label{eq:rains}
\end{equation}
Although $R\leq E_R$, the Rains bound need not upper bound AREEP, as
there are states for which $R(\rho)<E_R^\infty(\rho)$
\cite{WangDuan2017}.

\paragraph*{Exact PPT entanglement cost.}
A PPT operation is a completely positive, trace-preserving map
$\Lambda$ for which
$\Gamma\circ\Lambda\circ\Gamma$
is also completely positive. Let
$\Phi_2=|\phi_+\rangle\langle\phi_+|$, where
$|\phi_+\rangle=(|00\rangle+|11\rangle)/\sqrt2$.
A rate $r\geq0$ is achievable for exact PPT preparation of $\rho$
if, for every sufficiently large $n$, there is a PPT operation
$\Lambda_n$ such that
\begin{equation}
 \Lambda_n\!\left(\Phi_2^{\otimes\lfloor rn\rfloor}\right)
 =\rho^{\otimes n}.
 \label{eq:exact-preparation}
\end{equation}
The infimum of achievable rates is the exact PPT entanglement cost,
$E_{c,\mathrm{PPT}}^{\mathrm{exact}}(\rho)$
\cite{AudenaertPlenioEisert2003,LamiMeleRegula2025}.
Note that here the preparation error is required to vanish exactly, rather than merely vanishing
asymptotically. AREEP is known to lower bound the cost \cite{hayashi2017quantum}, i.e., 
\begin{equation}
 E_R^\infty(\rho)\leq E_{c,\mathrm{PPT}}^{\mathrm{exact}}(\rho).
 \label{eq:areep-cost}
\end{equation}

The logarithmic negativity is $E_N(\rho):=\log\|\rho^\Gamma\|_1$.
For states, which are not binegative, it equals the exact PPT
entanglement cost~\cite{AudenaertPlenioEisert2003}:
\begin{equation}
 |\rho^\Gamma|^\Gamma\ge0
 \quad\Longrightarrow\quad
 E_{c,\mathrm{PPT}}^{\mathrm{exact}}(\rho)=E_N(\rho).
 \label{eq:zero-binegativity}
\end{equation}
This positivity condition holds for every two-qubit
state~\cite{Ishizaka_2004}.

\paragraph*{The $\chi$- and $\kappa$-hierarchies.}
In~\cite{LamiMeleRegula2025} the authors introduced two
hierarchies of semidefinite programs. For $p\geq0$ and $q\geq1$,
with $S_{-1}=\rho$, they are given by
\begin{align}
    \chi_p(\rho) = \min &\{ \Tr S_p\,:\, -S_i \leq S_{i-1}^\Gamma \leq S_i \nonumber \\
    &\text{ for } i = 0,\ldots,p \text{ and } S_{-1} = \rho \}, \label{eq:chi_hierarchy}\\
    \kappa_q(\rho) = \min &\{ \Tr S_{q-1}\,:\, -S_i \leq S_{i-1}^\Gamma \leq S_i, \; S_{q-1}^\Gamma \geq 0 \nonumber \\
    &\text{ for } i = 0,\ldots,q-1 \text{ and } S_{-1} = \rho\}. \label{eq:kappa_hierarchy}
\end{align}
Set $E_{\chi,p}:=\log\chi_p$ and $E_{\kappa,q}:=\log\kappa_q$.
The lower bounds $E_{\chi,p}$ increase with $p$, the upper bounds
$E_{\kappa,q}$ decrease with $q$, and
\begin{align}
 E_{\chi,p}(\rho)
 &\leq E_{c,\mathrm{PPT}}^{\mathrm{exact}}(\rho)
 \leq E_{\kappa,q}(\rho),
 \label{eq:cost-bracket}\\
 \lim_{p\to\infty}E_{\chi,p}(\rho)
 &=E_{c,\mathrm{PPT}}^{\mathrm{exact}}(\rho)
 =\lim_{q\to\infty}E_{\kappa,q}(\rho).
 \label{eq:cost-limits}
\end{align}
In particular, $E_{\chi,0}=E_N$. If two consecutive $\chi$-levels give the same value fo some state, both hierarchies
collapse to its exact cost. In \cite{LamiMeleRegula2025}, the authors conjectured [Conjecture~S33]
\begin{equation}
 \chi_2(\rho)=\chi_3(\rho)
 \qquad\text{for all states }\rho,
 \label{eq:collapse-conjecture}
\end{equation}
which would imply
\begin{equation}
 E_{c,\mathrm{PPT}}^{\mathrm{exact}}(\rho)
 =E_{\chi,2}(\rho)=E_{\kappa,3}(\rho)
 \label{eq:conjectured-cost-formula}.
\end{equation}

\section{Upper bound hierarchy on $\AREEP$}
In our search for single-letter formulas, we now turn to the first main result of this Letter:
\begin{align}\label{eq:main_result1}
\AREEP\le \Tilde{E}^{1}_{N,k}\le \dots \le \Tilde{E}^{1}_{N,0} = E_R.
\end{align}
Here, the quantities $\Tilde{E}^{1}_{N,k}$ are computable one-copy convex optimizations, yielding a hierarchy of progressively tighter upper bounds on $\AREEP$. To the best of our knowledge, these constitute the first nontrivial single-copy upper bounds on this quantity.

We prove Eq.~\eqref{eq:main_result1} in the following. In order to define $\Tilde{E}^{1}_{N,k}$, it is useful to define supersets of the subnormalized cone of PPT states. Motivated from \eqref{eq:chi_hierarchy} and \eqref{eq:kappa_hierarchy}, we define
\begin{align}\label{eq:PPTk}
  \PPTk =
&\{
\omega_0 \ge 0 :
\exists\, \omega_1,\ldots,\omega_k
\ \text{s.t.} \ \Tr \omega_k \le 1, \nonumber \\  
&-\omega_{i+1}
\le
\omega_i^{\Gamma}
\le
\omega_{i+1}
\ \forall\, i = 0,\dots,k-1
\}.
\end{align}
These were already introduced in \cite{WangJingZhu2025}. $\text{PPT}_0$ is the set of all subnormalised states and $\text{PPT}_1$ just corresponds to the Rains set used in \eqref{eq:rains}. 

Similarly define the $\kappa$-like sets \cite{LamiMeleRegula2025}
\begin{align}\label{eq:PPTktilde}
  \PPTktilde =
&\{
\omega_0 \ge 0 :
\exists\, \omega_1,\ldots,\omega_k
\ \text{s.t.} \ \Tr \omega_k \le 1, \; \omega_k^\Gamma \ge 0, \nonumber \\  
&-\omega_{i+1}
\le
\omega_i^{\Gamma}
\le
\omega_{i+1}
\ \forall\, i = 0,\dots,k-1
\}.
\end{align}
Note that $\PPTktilde$ differs from $\PPTk$ only by the additional positive partial transpose constraint $\omega_k^\Gamma \ge 0$. For example, $\widetilde{\PPT}_0$ corresponds to the subnormalized PPT cone. The following inclusions are shown in the Supplemental Material \ref{app:set_inclusions}
\begin{align}\label{eq:PPTk_inclusions}
    &\widetilde{\PPT}_0 \subseteq  \widetilde{\PPT}_1 \subseteq  \dots \subseteq \widetilde{\PPT}_\infty\\ 
    =& \PPT_\infty \subseteq \ldots \subseteq \PPT_1 \subseteq \PPT_0.
\end{align}
A schematic illustration is given in Fig.~\ref{fig:PPT_k}.
\begin{figure}
    \centering
    \includegraphics[width=1.0\linewidth]{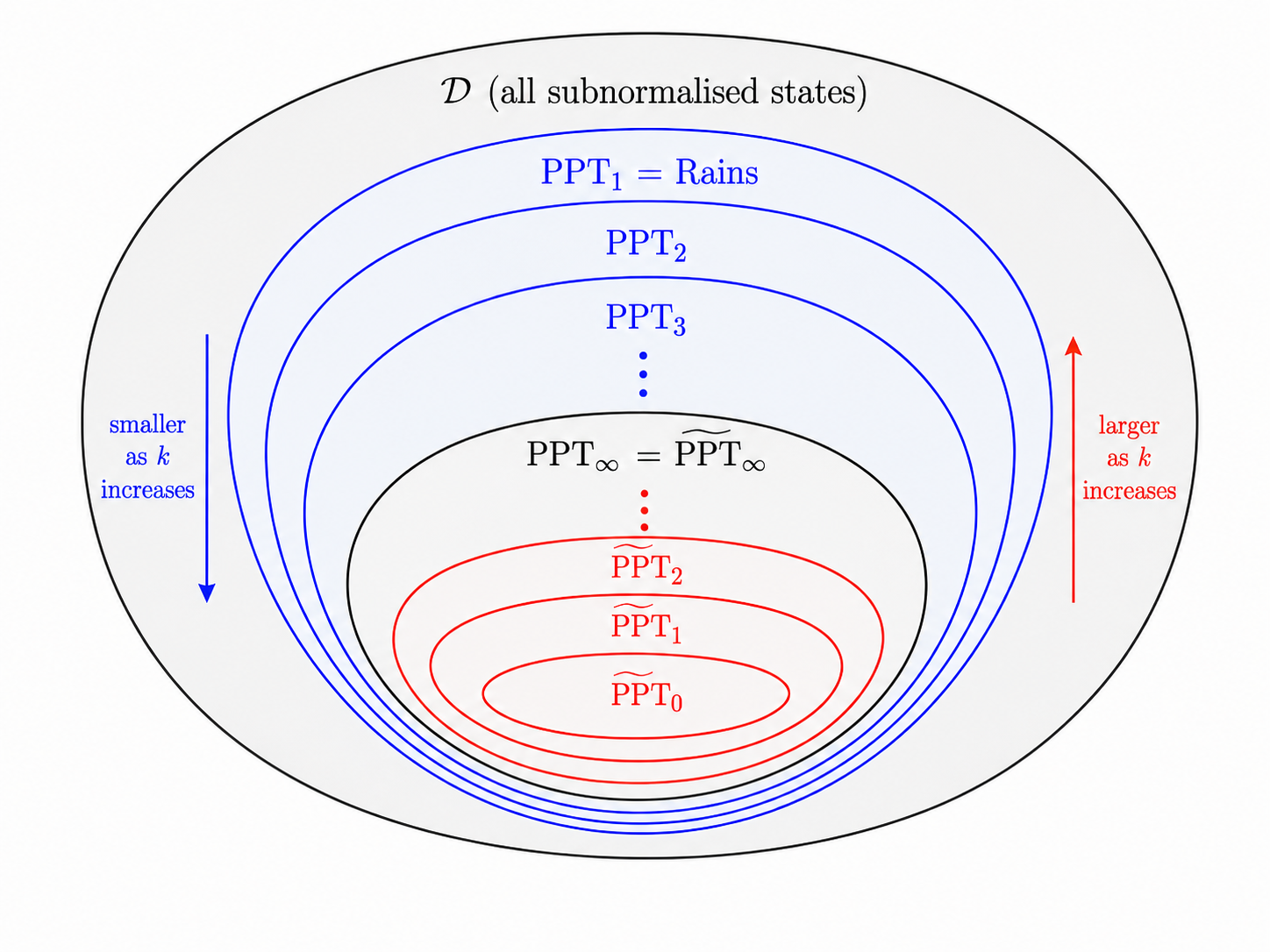}
    \caption{Schematic illustration of the set hierarchies $\mathrm{PPT}_k$ and $\widetilde{\mathrm{PPT}}_k$ defined in equations \eqref{eq:PPTk} and \eqref{eq:PPTktilde}. As $k$ increases, $\mathrm{PPT}_k$ decreases while $\widetilde{\mathrm{PPT}}_k$ increases, and both converge to the common limiting set $\mathrm{PPT}_\infty=\widetilde{\mathrm{PPT}}_\infty$.}
    \label{fig:PPT_k}
\end{figure}

It is interesting to note, that the "sandwich constraints" in \eqref{eq:PPTk} and \eqref{eq:PPTktilde} are tensor stable, i.e.
\begin{align}\label{eq:tensor_stable}
    -A \le B \le A \; , -C \le D \le C \nonumber \\
    \implies - A \otimes C  \le B \otimes D  \le A \otimes C , 
\end{align}
which was already proven in \cite{fangfawzi}, making $\PPTk$ closed under tensor products.

We can now define the following convex optimization over the set $\PPTktilde$, namely
\begin{align}\label{eq:E_Nk_tilde1}
    \Tilde{E}^{1}_{N,k}(\rho) := \min_{\sigma \in \PPTktilde} D(\rho \Vert \sigma).
\end{align}
Note, that $\Tilde{E}^{1}_{N,0}(\rho) = E_R(\rho)$.

\begin{theorem}\label{theorem1}
    $\Tilde{E}^{1}_{N,k}$, as defined in equation \eqref{eq:E_Nk_tilde1}, yields an upper bound on the asymptotic relative entropy of entanglement
\begin{align}\label{eq:theorem1}
    \AREEP\le \Tilde{E}^{1}_{N,k}.
\end{align}
\end{theorem}
The proof is given in the End Matter. As the sets $\PPTktilde$ become larger with increasing $k$ and the objective in \eqref{eq:E_Nk_tilde1} is a minimization, the upper bounds on $\AREEP$ can only improve, proving the first main result \eqref{eq:main_result1}.

Equation \eqref{eq:E_Nk_tilde1} exhibits several nice properties. First, it provides a single-letter formula and therefore does not require optimization over multiple copies of the state. Moreover, the feasible set is semidefinite-program representable, and the resulting problem is a convex optimization problem, making its computation tractable. Although it initially provides only upper bounds on AREEP, it nevertheless offers a new way to computing this quantity exactly in certain cases.
\begin{lemma}
    The regularization of $\Tilde{E}^{1}_{N,k}$ yields AREEP,
    \begin{align}
        E_R^\infty(\rho)&=\lim_{n\to\infty}\frac1n
                         \Tilde{E}^{1}_{N,k}(\rho^{\otimes n}). 
    \end{align}
\end{lemma}
\textit{Proof:} We have
\begin{align}
    \AREEP (\rho) \le \Tilde{E}^{1}_{N,k} (\rho) \le E_R (\rho). 
\end{align}
Regularizing this yields
\begin{align}
    \AREEP (\rho) \le \lim_{n\to\infty}\frac1n \Tilde{E}^{1}_{N,k}(\rho) \le \AREEP(\rho),
\end{align}
which proves the claim. $\Box$

If, for a given state $\rho$, $\Tilde{E}^{1}_{N,k}$ is additive at some level $k$, then
$\AREEP(\rho)=\Tilde{E}^{1}_{N,k}(\rho).$ Importantly, this may occur even when the basic upper bound $E_R=\Tilde{E}^{1}_{N,0}$ is nonadditive. 

We find two classes of states, for which exactly this happens. Both classes contain all Werner states. 
\paragraph*{Canonical symmetric-antisymmetric family.}
Let $d\geq2$, and for $0\leq i<j\leq d-1$, take the symmetric and antisymmetric states
\begin{align}
  |s_{ij}\rangle
  &:=\frac{|ij\rangle+|ji\rangle}{\sqrt{2}},
  &
  S_{ij}
  &:=|s_{ij}\rangle\langle s_{ij}|,
  \label{eq:main_supp-sij}\\
  |a_{ij}\rangle
  &:=\frac{|ij\rangle-|ji\rangle}{\sqrt{2}},
  &
  A_{ij}
  &:=|a_{ij}\rangle\langle a_{ij}|,
  \label{eq:main_supp-aij}\\
  D_i&:=|ii\rangle\langle ii|.
\end{align}
The first family consists of convex combinations of these basis states
\begin{equation}
  \rho_{\text{CSA}}
  =
  \sum_{i=0}^{d-1}p_iD_i
  +
  \sum_{0\leq i<j\leq d-1}
  \bigl(p_{ij}^{+}S_{ij}+p_{ij}^{-}A_{ij}\bigr),
  \label{eq:main_supp-rho-family}
\end{equation}
where all coefficients are nonnegative and
\begin{equation}
  \sum_{i=0}^{d-1}p_i
  +
  \sum_{0\leq i<j\leq d-1}
  (p_{ij}^{+}+p_{ij}^{-})
  =1.
\end{equation}

\paragraph*{Hyperoctahedral family.}
Consider the following projectors
\begin{align}
    &\Pi_0 = \ketbra{\Phi_d}{\Phi_d}, \\
    &\Pi_1 = \frac{\id + F}{2} - \Delta,  \\
    &\Pi_2 = \frac{\id - F}{2},\\
    & \Pi_3 = \Delta - \ketbra{\Phi_d}{\Phi_d},
\end{align}
with $\ketbra{\Phi_d}{\Phi_d}$ being the $d$-dimensional maximally entangled state, $F$ the swap operator $F = \sum_{ij=1}^d \ketbra{ji}{ij}$ and $\Delta := \sum_i D_i$. In \cite{nechita2026random} the following class of states was introduced:
\begin{align}\label{eq:main_HO_family}
    \rho_{\text{HO}} = \sum_{i=0}^{3} p_i \frac{\Pi_i}{\Tr \Pi_i}.
\end{align}

\begin{lemma}\label{lemma:family}
    For all $\rho_{\text{CSA}}$ and $\rho_{\text{HO}}$ of the form \eqref{eq:main_supp-rho-family} and \eqref{eq:main_HO_family}, respectively, $\AREEP$ can be calculated exactly by the convex optimization $\Tilde{E}^{1}_{N,1}$, i.e.
    \begin{align}
        \AREEP(\rho_{\text{CSA/HO}}) = \Tilde{E}^{1}_{N,1}(\rho_{\text{CSA/HO}}).
    \end{align}
\end{lemma}
This gives a single letter formula for AREEP for both multiparameter families \eqref{eq:main_supp-rho-family} and \eqref{eq:main_HO_family}. We show this in the Supplemental Material \ref{app:exact_areep_family}, by matching a known lower bound on $\AREEP$ from \cite{fangfawzi} to the newly established upper bound for this class of states. To the best of our knowledge, they are the broadest multiparameter families for which $\AREEP$ is known to admit an exact single-copy characterization.

Coming back to the bound from Theorem \ref{theorem1}, it turns out, that the same proof extends directly to the sandwiched R\'enyi divergences $ D_\alpha$ for all $\alpha\in[1/2,\infty)$ \cite{muller2013quantum,wilde2014strong,frank2013monotonicity}. 
Defining
\begin{align}
D_\alpha(\rho\Vert\sigma)
&:=
\frac{1}{\alpha-1}
\log
\operatorname{Tr}\!\left[
\left(
\sigma^{\frac{1-\alpha}{2\alpha}}
\rho\,
\sigma^{\frac{1-\alpha}{2\alpha}}
\right)^\alpha
\right],
\qquad \nonumber \\
&\alpha\in(0,1)\cup(1,\infty).\\
 E_{R,\alpha}^{\infty}(\rho)
&:=\lim_{n\to\infty}\frac{1}{n}\min_{\tau\in\PPT}
 D_\alpha(\rho^{\otimes n}\Vert\tau),\\
\qquad
\widetilde E^\alpha_{N,k}(\rho)
&:=\min_{\sigma\in\widetilde{\PPT}_k}
 D_\alpha(\rho\Vert\sigma),
\end{align}
the same proof as in Theorem \ref{theorem1} yields
\begin{align}
E_{R,\alpha}^{\infty}(\rho)
\leq \widetilde E^\alpha_{N,k}(\rho), \quad \forall \alpha \in [1/2,\infty).
\end{align}
Indeed, besides of the tensor-stable construction of the PPT trial state $T_n$, the proof only uses additivity under tensor products, the scaling relation
$ D_\alpha(\rho\Vert c\sigma)= D_\alpha(\rho\Vert\sigma)-\log c$, and antimonotonicity in the second argument. Hence the same estimate gives an additive $\log(k+1)$ correction, which vanishes upon regularization.

\section{k-multinegative states}
Although \eqref{eq:E_Nk_tilde1} is already a single copy convex optimisation, we can find an even easier to compute upper bound by constructing a feasible point in $\PPTktilde$. For this, we first define the \textit{absolute partial transposition} $\mathcal{N}$ and the \textit{k-fold absolute partial transposition} $\mathcal{N}_k$ of an operator $X$ by
\begin{align}\label{eq:Nmap}
    \mathcal{N}(X) = \abs{X}^\Gamma, \quad \mathcal{N}_k(X) = \underbrace{\mathcal{N} \circ \mathcal{N} \circ \dots \circ \mathcal{N}}_{k \; \text{times}} \;(X)  .
\end{align}

Applying $\mathcal{N}$ on a quantum state once just corresponds to the ordinary partial transposition $\mathcal{N}_1(\rho) = \rho^\Gamma$, whereas applying it twice gives the operator $\mathcal{N}_2(\rho) = \abs{\rho^\Gamma}^\Gamma$. This motivates the definition of \textit{$k$-multinegative} states.
\begin{definition}
    A quantum state $\rho$ is called \textit{$k$-multinegative}, if $\mathcal{N}_k(\rho) \not \ge 0$, where $\mathcal{N}_k$ is defined in equation \eqref{eq:Nmap}.
\end{definition}

It is already known, that $E_N=\log \Tr \mathcal{N}_2(\rho)$ is a lower bound on exact entanglement cost under PPT operations $E_{\text{c,PPT}}^{\text{exact}}$ \cite{AudenaertPlenioEisert2003}, and that for states, which are not binegative ($\mathcal{N}_2(\rho) \ge 0$), it is an upper bound for $\AREEP$ \cite{ADVW2002}. This can be generalised by the following Lemma.
\begin{lemma}\label{lemma:k-multinegativity}
    The k-multinegativity $\log \Tr \mathcal{N}_k(\rho)$ upper bounds $\AREEP$ and $E_{\text{c,PPT}}^{\text{exact}}$, if ${\mathcal{N}}_k(\rho) \ge 0$, i.e.
    \begin{align}
        {\mathcal{N}}_k(\rho) \ge 0 \implies \AREEP(\rho) \le E_{\text{c,PPT}}^{\text{exact}}(\rho) \le \log \Tr \mathcal{N}_k(\rho).
    \end{align}
\end{lemma}

\begin{proof}
    Consider $k \ge 2$. We show that, $\Tr \mathcal{N}_k(\rho)$ upper bounds $\kappa_{k-1}(\rho)$ from Eq.~\eqref{eq:kappa_hierarchy}, if ${\mathcal{N}}_k(\rho) \ge 0 $.  For this, construct the feasible set of operators $\{ S^\ast_i \}_{i=0}^{k-2}$, with
    \begin{align}
    S^\ast_{i} &= \abs{S_{i-1}^{\ast\Gamma}} \quad \forall \; i \in \{0,\dots, k-2\}, \; \; S^\ast_{-1} = \rho.
    \end{align}
    It is easy to verify, that these operators fulfill the sandwich constraints in Eq.~\eqref{eq:kappa_hierarchy}. Furthermore, if ${\mathcal{N}}_k(\rho) \ge 0 $, the PPT constraint $S_{k-2}^{\ast\Gamma} \ge 0$ is fulfilled as well. The set $\{ S^\ast_i \}_{i=0}^{k-2}$ therefore forms a feasible point for $\kappa_{k-1}(\rho)$. This implies, that $\Tr S_{k-2}^{\ast} = \Tr \mathcal{N}_k(\rho) \ge \kappa_{k-1}(\rho)$. Since $\log \kappa_{k-1}(\rho)$ upper bounds $E_{\text{c,PPT}}^{\text{exact}}$, which itself upper bounds $\AREEP$, Lemma \ref{lemma:k-multinegativity} follows.
\end{proof}

Apart from Lemma \ref{lemma:k-multinegativity}, which gives $k$-multinegativity a partial operational interpretation by providing an upper bound on $\AREEP$ and $E_{\text{c,PPT}}^{\text{exact}}$, it is not clear whether multinegative states have any further physical relevance. Nonetheless, they do appear to be useful as counterexamples to a conjecture stated in Ref.~\cite{LamiMeleRegula2025} (Conjecture~S33). They conjecture a collapse of both hierarchies, or more explicitly that
\begin{align}
E_{\text{c,PPT}}^{\text{exact}}(\rho) \overset{?}{=} E_{\chi,2}(\rho) \overset{?}{=} E_{\kappa,3}(\rho),
\end{align}
and additionally state the stronger conjecture that $E_{\chi,2}(\rho) = E_{\kappa,2}(\rho)$, since no counterexamples were found. It turns out that certain $k$-multinegative states can disprove this conjecture and stronger versions of it.

The intuition for considering such states is the following. In order to observe a difference between two levels of the $\kappa$-hierarchy, the optimizer $\tau$ must satisfy $\tau \in \widetilde{\PPT}_{q}$ while $\tau \notin \widetilde{\PPT}_{q-1}$. Since the property $\mathcal{N}_q(\tau) \ge 0$ is a sufficient condition for a suitably normalized $\tau$ to belong to $\widetilde{\PPT}_{q-1}$, as shown in Lemma \ref{lemma:k-multinegativity}, the optimizer must have sufficiently high $k$-multinegativity in order for a gap to appear at a higher level of the hierarchy. It therefore seems plausible to consider a state $\rho$ that is itself $k$-multinegative.

In the Supplemental Material, we present for each choice of $k\geq 3$ a counterexample $\rho_k$, for which
\begin{align}
   \chi_{k-1}(\rho_k)&=2^k \binom{k}{\lfloor k/2\rfloor}^{-1}, \\
 \chi_{k-2}(\rho_k) &\leq (2^k-1)\binom{k}{\lfloor k/2\rfloor}^{-1},
\end{align}
such that $\chi_{k-1}(\rho_k) - \chi_{k-2}(\rho_k) \geq \binom{k}{\lfloor k/2\rfloor}^{-1} > 0$.
This shows a difference in two levels of the hierarchy for arbitrary $k$. Furthermore, $\rho_k$ is $k$-multinegative. Thus, the family of counterexamples yields examples of state with arbitrarily high levels of multinegativity. It is interesting to note, that these states do not show a non-collapse in the $\kappa$-hierarchy. It is therefore still an open problem, if the $\kappa$-hierarchy exhibits a finite single letter formula for exact entanglement cost.

\section{Discussion and Conclusions}

Our results clarify both the possibilities and the limitations of single-copy approaches to asymptotic entanglement quantities under PPT operations. On the one hand, we construct a hierarchy that provides a computable one-copy upper bound on the asymptotic relative entropy of entanglement w.r.t.~PPT states (AREEP). It relies only on tensor-stable semidefinite constraints and therefore bypasses the explicit many-copy regularization appearing in the definition of AREEP. We demonstrated that already the first nontrivial level can be strictly tighter than the one-copy PPT relative entropy of entanglement, and we find large multiparameter families in \eqref{eq:main_supp-rho-family} and \eqref{eq:main_HO_family} for which it is in fact exact. Thus, although regularization cannot in general be discarded by a naive one-copy replacement, nontrivial single-copy characterizations remain possible for broad classes of states.

At the same time, our counterexamples show that an analogous hierarchy for the exact PPT entanglement cost defined in Ref.~\cite{LamiMeleRegula2025} do not show a finite collapse, contrary to prior belief. In particular, we constructed a family of $k$-multinegative states, i.e., states which remain negative after repeated application of matrix absolute value and partial transposition, which exhibit a gap in that hierarchy for arbitrarily high order. Consequently, no fixed finite level of this hierarchy suffices for all quantum states. This does not exclude the possibility that the exact PPT cost admits a different single-letter characterization.

The iterated absolute partial transpose turns out to be a useful tool. For states, which are not binegative, a simple expression in terms of logarithmic negativity exists. Otherwise, when the iterated absolute partial transpose turns positive at some level, its logarithmic trace gives an explicit upper bound on both the AREEP and exact PPT cost. 

Several questions concerning the upper-bound hierarchy remain open. In particular, it is natural to ask how tightly it approximates AREEP for generic states, whether further families can be identified for which a finite level is exact, and how the hierarchy behaves as its level increases. The extension of the construction to sandwiched R\'enyi divergences for all $\alpha\geq 1/2$ suggests that the underlying tensor-stable mechanism is not specific to the ordinary relative entropy. Exploring analogous constructions for other regularized entanglement quantities may therefore provide a broader route to tractable one-copy bounds even in settings where an exact finite single-letter formula is unavailable.

\textit{Note added.---} While completing this work, we became aware of related work in \cite{ao2026ppt}, where by a different method and proof the authors found an upper bound on AREEP, which equals the upper bound in \eqref{eq:theorem1} for $k \to \infty$. 

\begin{acknowledgments}
The numerical search for multinegative states and the identification of candidate state families were assisted by ChatGPT (OpenAI, GPT-5.6 Sol; accessed August 2026). The initial draft of the paper was written by hand and iteratively improved with the help of ChatGPT 5.6 Sol, accessed August 2026. All AI-assisted material was
reviewed and revised by the authors, who take full responsibility for the
content. This work has been supported by the Federal
Ministry of Research, Technology and Space (BMFTR Projects QR.N, Grant
No.~16KIS2202 and QSolid, Grant No.~13N16163). We also acknowledge
financial support by Deutsche Forschungsgemeinschaft (DFG, German Research
Foundation) under Germany's Excellence Strategy -- Cluster of Excellence
Matter and Light for Quantum Computing (ML4Q) EXC 2004/1 -- 390534769.
\end{acknowledgments}

\bibliography{aps}

\section*{End matter}

We prove Theorem \ref{theorem1}, i.e. that $\AREEP \le \Tilde{E}^{1}_{N,k}$ for all $k$.

\begin{proof}
Take a feasible point $\omega_0 \in \PPTktilde$. From \eqref{eq:PPTktilde} we know a set of operators $\{\omega_{i} \}_{i=1}^k$ exists, which fulfill
\begin{align}
    \Tr\omega_i \le 1 \label{eq:trace_omega_i} \\
    -\omega_{i+1} \le \omega_i^\Gamma \le \omega_{i+1}, \\
    \omega_k^\Gamma \ge 0, \label{eq:omega_k_ppt}
\end{align}
for $i = 0,\ldots,k-1$.
Because of tensor stability (Eq.~\eqref{eq:tensor_stable}) we furthermore have
\begin{align}
    \Tr\omega_i^{\otimes n} \le 1, \\
    -\omega_{i+1}^{\otimes n} \le (\omega_i^\Gamma)^{\otimes n} \le \omega_{i+1}^{\otimes n}. \label{eq:sandwich_n}
\end{align}
We can now construct the (unnormalised) trial state for $E_R(\rho^{\otimes n})$, namely
\begin{align*}
    T_n = \underbrace{\omega_0^{\otimes n}}_{\ge 0} + \underbrace{(\omega_1^\Gamma)^{\otimes n} +\omega_2^{\otimes n}}_{\ge 0}  + \dots + \begin{cases} (\omega_k^\Gamma)^{\otimes n}, & \text{$k$ odd,} \\ \omega_k^{\otimes n}, & \text{$k$ even.}\end{cases}
\end{align*}
$T_n$ is positive by construction, as $\omega_0$ is positive and then every two consecutive terms are positive by the sandwich constraint \eqref{eq:sandwich_n}. If $k$ is odd, the last remaining term $(\omega_k^\Gamma)^{\otimes n}$ is still positive by \eqref{eq:omega_k_ppt}. By the same argument $T_n$ can be shown to be PPT
\begin{align*}
    T_n^\Gamma = \underbrace{(\omega_0^\Gamma)^{\otimes n} + \omega_1^{\otimes n}}_{\ge 0}  + \dots + \begin{cases} \omega_k^{\otimes n}, & \text{$k$ odd}, \\ (\omega_k^\Gamma)^{\otimes n}, & \text{$k$ even.}\end{cases}
\end{align*}
Therefore, $t_n = T_n/\Tr(T_n)$ can be taken as a trial state in $E_R(\rho^{\otimes n})$,
\begin{align}
    E_R(\rho^{ \otimes n}) &\le D(\rho^{ \otimes n} \Vert t_n ) = D(\rho^{ \otimes n} \Vert T_n) + \log \Tr T_n \\
    &\le D(\rho^{ \otimes n} \Vert T_n) + \log (k+1)  \\
    &\le D(\rho^{ \otimes n} \Vert \omega_0^{\otimes n} ) + \log (k+1) \\
    & = nD(\rho\Vert \omega_0) + \log (k+1) .
\end{align}
The second inequality stems from \eqref{eq:trace_omega_i} and the third inequality follows from anti-monotonicity of the relative entropy ($T_n\ge \omega_0^{\otimes n} \implies D(\rho^{ \otimes n} \Vert T_n) \le D(\rho^{ \otimes n} \Vert \omega_0^{\otimes n} )$).
Regularizing (i.e., dividing by $n$ and taking the limit) then yields
\begin{align}
    \AREEP \le D(\rho\Vert \omega_0) \; \forall \; \omega_0 \in \PPTktilde.
\end{align}
Minimizing the right-hand side over $\omega_0\in\PPTktilde$ proves \eqref{eq:theorem1}. \end{proof}

\clearpage

\appendix
\onecolumngrid

\section{$\PPTk$ and $\PPTktilde$ set inclusions}\label{app:set_inclusions}

The sets $\PPTk$ and $\PPTktilde$ are defined in \eqref{eq:PPTk} and \eqref{eq:PPTktilde}. The inclusions
\begin{equation}
    \mathrm{PPT}\subseteq \mathrm{PPT}_k\subseteq\cdots
    \subseteq \mathrm{PPT}_1\subseteq\mathrm{PPT}_0
\end{equation}
have already been established in Ref.~\cite{WangJingZhu2025}, so it remains to verify the corresponding statements involving $\widetilde{\mathrm{PPT}}_k$. We claim that
\begin{equation}
    \widetilde{\mathrm{PPT}}_k
    \subseteq
    \widetilde{\mathrm{PPT}}_{k+1}.
\end{equation}
Indeed, let $\omega_0\in\widetilde{\mathrm{PPT}}_k$ with witnesses
$\omega_1,\ldots,\omega_k$. The sandwich constraints imply
$\omega_i\geq0$ for every $i$, while by definition
$\omega_k^\Gamma\geq0$. Setting
$\omega_{k+1}:=\omega_k^\Gamma$, we have
$\operatorname{Tr}\omega_{k+1}
=\operatorname{Tr}\omega_k\leq1$,
$(\omega_{k+1})^\Gamma=\omega_k\geq0$, and
\[
    -\omega_{k+1}
    \leq \omega_k^\Gamma
    \leq \omega_{k+1},
\]
so the extended tuple is feasible for
$\widetilde{\mathrm{PPT}}_{k+1}$. This proves
$\widetilde{\mathrm{PPT}}_0
\subseteq\widetilde{\mathrm{PPT}}_1
\subseteq\cdots\subseteq\widetilde{\mathrm{PPT}}_k$.

To see $\widetilde{\PPT}_\infty = \PPT_\infty$, note the following relation to the $\chi$-hierarchy from \cite{LamiMeleRegula2025}, which was already established \cite{WangJingZhu2025}
\begin{align}
    X \in \PPT_k \iff \chi_{k-1}(X) \le 1.
\end{align}
Analogously, by adding the PPT constraint $\omega_k^\Gamma \ge 0$, one arrives at
\begin{align}
    X \in \widetilde{\PPT}_k \iff \kappa_{k}(X) \le 1.
\end{align}
Since we know both hierarchies converge to the same value for $k \to \infty$ \cite{LamiMeleRegula2025}, we get
\begin{align}
    X \in \PPT_\infty \iff 2^{E_{c,\mathrm{PPT}}^{\mathrm{exact}}(X)} \le 1 \iff X \in \widetilde{\PPT}_\infty.
\end{align}
This proves equation \eqref{eq:PPTk_inclusions}.

\section{Exact $\AREEP$ for families of states}\label{app:exact_areep_family}

We start with the canonical symmetric-antisymmetric family \eqref{eq:main_supp-rho-family}. We prove that every state in this family satisfies
\begin{equation}
   \inf_{\tau\in\mathrm{PPT}_2}D_{\mathbb M}(\rho_{\text{CSA}}\Vert \tau)
  =
  E_R^\infty(\rho_{\text{CSA}})
  =
  \inf_{\tau\in\widetilde{\mathrm{PPT}}_1}
  D(\rho_{\text{CSA}}\Vert\tau).
  \label{eq:supp-main-result}
\end{equation}
Recall from Eqs.~\eqref{eq:PPTk} and \eqref{eq:PPTktilde} the definitions of the sets $\mathrm{PPT}_2$ and $ \widetilde{\mathrm{PPT}}_1$. Using the variational form of the trace norm, these sets can be written equivalently as
\begin{equation}
  \mathrm{PPT}_2
  :=
  \left\{
    X\geq0:
    \begin{array}{l}
      \text{there exists }Y\geq0\text{ such that}\\[-1mm]
      -Y\leq X^\Gamma\leq Y,\quad
      \|Y^\Gamma\|_1\leq1
    \end{array}
  \right\},
  \label{eq:supp-ppt2}
\end{equation}
and
\begin{equation}
  \widetilde{\mathrm{PPT}}_1
  :=
  \left\{
    X\geq0:
    \begin{array}{l}
      \text{there exists }Y\geq0\text{ such that}\\[-1mm]
      -Y\leq X^\Gamma\leq Y,\quad
      \operatorname{Tr}Y\leq1,\quad
      Y^\Gamma\geq0
    \end{array}
  \right\}.
  \label{eq:supp-ppt1-tilde}
\end{equation}
The measured relative entropy is given by \cite{fangfawzi}
\[
  D_{\mathbb M}(\rho\Vert\tau)
  :=
  \sup_{\mathcal M}
  D\bigl(\mathcal M(\rho)\Vert\mathcal M(\tau)\bigr),
\]
where the supremum is over all POVMs $\mathcal{M}$.
We use the lower bound from \cite{fangfawzi} and the upper bound from Eq.~\eqref{eq:theorem1},
\begin{equation}
  \inf_{\tau\in\mathrm{PPT}_2}D_{\mathbb M}(\rho\Vert \tau)
  \leq
  E_R^\infty(\rho)
  \leq
  \inf_{\tau\in\widetilde{\mathrm{PPT}}_1}
  D(\rho\Vert\tau).
  \label{eq:supp-general-bounds}
\end{equation}
From Eq.~\eqref{eq:PPTk_inclusions}, we also have that
\begin{equation}
  \widetilde{\mathrm{PPT}}_1
  \subseteq
  \mathrm{PPT}_2.
  \label{eq:supp-cone-inclusion}
\end{equation}

We now turn to proving \eqref{eq:supp-main-result}.

\medskip
\noindent\textit{Step 1: the canonical twirl.}
Let
\begin{equation}
  \mathcal A_d
  :=
  \operatorname{span}
  \{D_i,S_{ij},A_{ij}:
    0\leq i\leq d-1,\ 0\leq i<j\leq d-1\}.
  \label{eq:supp-canonical-algebra}
\end{equation}
It is spanned by mutually
orthogonal projectors.  For
$\boldsymbol{\theta}=(\theta_0,\ldots,\theta_{d-1})$, define
\[
  U_{\boldsymbol{\theta}}
  :=
  \sum_{i=0}^{d-1}
  e^{\mathrm{i}\theta_i}|i\rangle\langle i|
\]
and
\begin{align}
  \mathcal T_0(X)
  &:=
  \int_{[0,2\pi)^d}
  \frac{\mathrm d^d\boldsymbol{\theta}}{(2\pi)^d}
  (U_{\boldsymbol{\theta}}\otimes U_{\boldsymbol{\theta}})
  X
  (U_{\boldsymbol{\theta}}\otimes U_{\boldsymbol{\theta}})^\dagger,
  \label{eq:supp-T0}\\
  \mathcal T(X)
  &:=
  \frac{1}{2}
  \bigl(\mathcal T_0(X)+F\mathcal T_0(X)F\bigr),
  \label{eq:supp-T}
\end{align}
where $F$ is the swap operator with $F|ij\rangle=|ji\rangle$.

We now prove explicitly that
\begin{equation}
  \operatorname{range}(\mathcal T)=\mathcal A_d,
  \qquad
  \mathcal T(\rho_{\text{CSA}})=\rho_{\text{CSA}}.
  \label{eq:supp-twirl-range}
\end{equation}
Expand an arbitrary operator as
\[
  X
  =
  \sum_{i,j,k,l=0}^{d-1}
  x_{ij,kl}|ij\rangle\langle kl|.
\]
Conjugation by
$U_{\boldsymbol{\theta}}\otimes U_{\boldsymbol{\theta}}$ multiplies
$|ij\rangle\langle kl|$ by
\[
  e^{\mathrm{i}(\theta_i+\theta_j-\theta_k-\theta_l)}.
\]
The integral in \eqref{eq:supp-T0} is nonzero precisely when
\[
  \delta_{ri}+\delta_{rj}
  =
  \delta_{rk}+\delta_{rl}
  \qquad\text{for every }r,
\]
or, equivalently, when the sets $\{i,j\}$ and $\{k,l\}$ are equal.
For each $i$, it retains only the one-dimensional block on
$\operatorname{span}\{|ii\rangle\}$, while for every $i<j$ it retains the
full $2\times2$ block on
\[
  \mathcal H_{ij}
  :=
  \operatorname{span}\{|ij\rangle,|ji\rangle\}.
\]
In the ordered basis $(|ij\rangle,|ji\rangle)$ of this block, the restriction of $F$ onto it can be written as $P=\left(\begin{smallmatrix}0&1\\1&0\end{smallmatrix}\right)$.  Therefore,
on $\mathcal H_{ij}$, the second average sends
\[
  B=
  \begin{pmatrix}a&b\\c&d\end{pmatrix}
  \quad\longmapsto\quad
  \frac{1}{2}(B+PBP)
  =
  \begin{pmatrix}
    (a+d)/2&(b+c)/2\\
    (b+c)/2&(a+d)/2
  \end{pmatrix}.
\]
The resulting matrix is a linear combination of the two eigenvector
projectors of $P$, which are precisely $S_{ij}$ and $A_{ij}$.  Hence
$\operatorname{range}(\mathcal T)\subseteq\mathcal A_d$.  Conversely,
every $D_i,S_{ij},A_{ij}$ is fixed by both averages, proving the reverse
inclusion.  This proves the first identity in
\eqref{eq:supp-twirl-range}.  Since every summand of $\rho_{\text{CSA}}$ is fixed,
the second identity follows as well.

\medskip
\noindent\textit{Step 2: invariance of the two cones and symmetry
reduction.}
Set
\[
  W_{\boldsymbol{\theta}}
  :=
  U_{\boldsymbol{\theta}}\otimes U_{\boldsymbol{\theta}},
  \qquad
  V_{\boldsymbol{\theta}}
  :=
  U_{\boldsymbol{\theta}}\otimes
  \overline{U}_{\boldsymbol{\theta}}.
\]
If $X'=W_{\boldsymbol{\theta}}XW_{\boldsymbol{\theta}}^\dagger$, then
\[
  (X')^\Gamma
  =
  V_{\boldsymbol{\theta}}X^\Gamma
  V_{\boldsymbol{\theta}}^\dagger.
\]
Consequently, if $Y$ is a feasible witness for $X$ in either
$\mathrm{PPT}_2$ or $\widetilde{\mathrm{PPT}}_1$, then
$V_{\boldsymbol{\theta}}YV_{\boldsymbol{\theta}}^\dagger$ is a feasible
witness for $X'$.  Positivity and the sandwich constraints are preserved,
and
\[
  \bigl(
    V_{\boldsymbol{\theta}}YV_{\boldsymbol{\theta}}^\dagger
  \bigr)^\Gamma
  =
  W_{\boldsymbol{\theta}}Y^\Gamma
  W_{\boldsymbol{\theta}}^\dagger.
\]
Thus the trace norm, trace, and PPT conditions are
preserved.

The swap also preserves both cones.  Indeed, if $X'=FXF$, then
\[
  (X')^\Gamma
  =
  F(X^\Gamma)^{\mathsf T}F.
\]
Therefore, if $Y$ witnesses feasibility of $X$, then $FY^{\mathsf T}F$ witnesses
feasibility of $X'$. It follows by convexity that $\mathcal T$ preserves both
$\mathrm{PPT}_2$ and $\widetilde{\mathrm{PPT}}_1$.

Since $\mathcal T(\rho_{\text{CSA}})=\rho_{\text{CSA}}$, the data processing inequality for $D$ and
$D_{\mathbb M}$ implies
\[
  D_*(\rho_{\text{CSA}}\Vert\tau)
  \geq
  D_*\bigl(
    \mathcal T(\rho_{\text{CSA}})\Vert\mathcal T(\tau)
  \bigr)
  =
  D_*(\rho_{\text{CSA}}\Vert\mathcal T(\tau)),
\]
where $D_*$ denotes either divergence.  Therefore
\begin{align}
  \inf_{\tau\in\mathrm{PPT}_2}D_{\mathbb M}(\rho_{\text{CSA}}\Vert \tau)
  &=
  \inf_{\tau\in\mathrm{PPT}_2\cap\mathcal A_d}D_{\mathbb M}(\rho_{\text{CSA}}\Vert \tau),
  \label{eq:supp-reduce-ppt2}\\
  \inf_{\tau\in\widetilde{\mathrm{PPT}}_1}
  D(\rho_{\text{CSA}}\Vert\tau)
  &=
  \inf_{\tau\in\widetilde{\mathrm{PPT}}_1\cap\mathcal A_d}
  D(\rho_{\text{CSA}}\Vert\tau).
  \label{eq:supp-reduce-ppt1}
\end{align}
Moreover, every $\tau\in\mathcal A_d$ commutes with $\rho_{\text{CSA}}$.
Measurement in their common eigenbasis therefore attains the quantum
relative entropy, so
\begin{equation}
  D_{\mathbb M}(\rho_{\text{CSA}}\Vert\tau)
  =
  D(\rho_{\text{CSA}}\Vert\tau),
  \qquad
  \tau\in\mathcal A_d,\quad \tau\geq0.
  \label{eq:supp-commuting}
\end{equation}

\medskip
\noindent\textit{Step 3: equality of the two cones on
$\mathcal A_d$.}
We claim that
\begin{equation}
  \mathrm{PPT}_2\cap\mathcal A_d
  =
  \widetilde{\mathrm{PPT}}_1\cap\mathcal A_d.
  \label{eq:supp-reduced-cone-equality}
\end{equation}
The inclusion from right to left follows from
\eqref{eq:supp-cone-inclusion}.  For the converse, take
$X\in\mathrm{PPT}_2\cap\mathcal A_d$.  There is a positive $Y$ satisfying
\begin{equation}
  -Y\leq X^\Gamma\leq Y,
  \qquad
  \|Y^\Gamma\|_1\leq1.
  \label{eq:supp-initial-witness}
\end{equation}
We next symmetrize this witness and derive its structure in detail.

For $i<j$, direct partial transposition gives
\begin{align*}
  D_i^\Gamma
  &=
  D_i,\\
  S_{ij}^\Gamma
  &=
  \frac{1}{2}
  \bigl(
    |ij\rangle\langle ij|
    +
    |ji\rangle\langle ji|
    +
    |ii\rangle\langle jj|
    +
    |jj\rangle\langle ii|
  \bigr),\\
  A_{ij}^\Gamma
  &=
  \frac{1}{2}
  \bigl(
    |ij\rangle\langle ij|
    +
    |ji\rangle\langle ji|
    -
    |ii\rangle\langle jj|
    -
    |jj\rangle\langle ii|
  \bigr).
\end{align*}
These operators are invariant under conjugation by
$V_{\boldsymbol{\theta}}$ and by $F$.  Hence $X^\Gamma$ has both
invariances.  Average $Y$ first under
$Y\mapsto V_{\boldsymbol{\theta}}YV_{\boldsymbol{\theta}}^\dagger$ and
then under $Y\mapsto FYF$.  The sandwich constraint in
\eqref{eq:supp-initial-witness} is preserved because $X^\Gamma$ is fixed.
Positivity is preserved, and the trace norm cannot increase by convexity.
For the swap average, one uses
\[
  (FYF)^\Gamma
  =
  F(Y^\Gamma)^{\mathsf T}F,
\]
which has the same trace norm as $Y^\Gamma$.

To see the form of the averaged witness directly, expand the original
witness as
\[
  Y
  =
  \sum_{i,j,k,l=0}^{d-1}
  c_{ij,kl}|ij\rangle\langle kl|.
\]
For every matrix $|ij\rangle\langle kl|$,
\[
  V_{\boldsymbol{\theta}}
  |ij\rangle\langle kl|
  V_{\boldsymbol{\theta}}^\dagger
  =
  e^{\mathrm{i}(\theta_i-\theta_j-\theta_k+\theta_l)}
  |ij\rangle\langle kl|.
\]
The phase average is zero unless the coefficient of every independent
phase \(\theta_r\) vanishes, that is,
\[
  \delta_{ri}-\delta_{rj}-\delta_{rk}+\delta_{rl}=0
  \qquad\text{for every }r.
\]
If \(i=j\),
this condition requires \(k=l\), so all terms
\[
  |ii\rangle\langle kk|
\]
survive.  If \(i\neq j\), the condition requires \(k=i\) and \(l=j\),
so only the diagonal term
\[
  |ij\rangle\langle ij|
\]
survives.  Thus, after the phase average, the witness has the form
\[
  Y_0
  =
  \sum_{i,j=0}^{d-1}
  Z_{ij}|ii\rangle\langle jj|
  +
  \sum_{i\neq j}
  r_{ij}|ij\rangle\langle ij|.
\]

We next average with respect to the swap.  Since
\[
  F|ii\rangle=|ii\rangle,
  \qquad
  F|ij\rangle=|ji\rangle,
\]
we obtain
\begin{align}
  \frac12(Y_0+FY_0F)
  ={}&
  \sum_{i,j=0}^{d-1}
  Z_{ij}|ii\rangle\langle jj|
  \notag\\
  &+
  \sum_{0\leq i<j\leq d-1}
  \frac{r_{ij}+r_{ji}}{2}
  \bigl(
    |ij\rangle\langle ij|
    +
    |ji\rangle\langle ji|
  \bigr).
\end{align}
Defining
\[
  y_{ij}:=\frac{r_{ij}+r_{ji}}{2}
\]
and relabeling the averaged witness as \(Y\), this becomes
\[
  Y
  =
  \sum_{i,j=0}^{d-1}
  Z_{ij}|ii\rangle\langle jj|
  +
  \sum_{0\leq i<j\leq d-1}
  y_{ij}
  \bigl(
    |ij\rangle\langle ij|
    +
    |ji\rangle\langle ji|
  \bigr).
\]
The first term is the restriction of \(Y\) to
\(\operatorname{span}\{|ii\rangle\}_{i=0}^{d-1}\), while all remaining
terms are one-dimensional diagonal blocks.  Since the averaged witness
is positive,
\[
  Z\geq0,
  \qquad
  y_{ij}\geq0.
\]

Partial transposition maps
$|ii\rangle\langle jj|$ to $|ij\rangle\langle ji|$ and leaves
$|ij\rangle\langle ij|$ unchanged.  Since
$Z_{ji}=\overline{Z_{ij}}$, it follows that
\begin{align}
  Y^\Gamma
  ={}&
  \sum_{i=0}^{d-1}
  Z_{ii}|ii\rangle\langle ii|
  \notag\\
  &+
  \sum_{0\leq i<j\leq d-1}
  \Bigl[
    y_{ij}
    \bigl(
      |ij\rangle\langle ij|
      +
      |ji\rangle\langle ji|
    \bigr)
    +
    Z_{ij}|ij\rangle\langle ji|
    +
    \overline{Z_{ij}}|ji\rangle\langle ij|
  \Bigr].
  \label{eq:supp-Y-Gamma-expanded}
\end{align}
Thus $Y^\Gamma$ is the direct sum of the one-dimensional blocks
$[Z_{ii}]$ on $|ii\rangle$ and, for each $i<j$, the block
\begin{equation}
  B_{ij}
  =
  \begin{pmatrix}
    y_{ij}&Z_{ij}\\
    \overline{Z_{ij}}&y_{ij}
  \end{pmatrix}
  \quad
  \text{on }
  \operatorname{span}\{|ij\rangle,|ji\rangle\}.
  \label{eq:supp-Y-Gamma-block}
\end{equation}
The eigenvalues of $B_{ij}$ are
\[
  y_{ij}+|Z_{ij}|
  \quad\text{and}\quad
  y_{ij}-|Z_{ij}|.
\]
Because $y_{ij}\geq0$,
\[
  \|B_{ij}\|_1
  =
  \bigl|y_{ij}+|Z_{ij}|\bigr|
  +
  \bigl|y_{ij}-|Z_{ij}|\bigr|
  =
  2\max\{y_{ij},|Z_{ij}|\}.
\]
The trace norm is additive over orthogonal direct sums, while
$Z_{ii}\geq0$.  Consequently,
\begin{equation}
  \|Y^\Gamma\|_1
  =
  \operatorname{Tr}Z
  +
  2
  \sum_{0\leq i<j\leq d-1}
  \max\{y_{ij},|Z_{ij}|\}.
  \label{eq:supp-Y-Gamma-norm}
\end{equation}

Define
\[
  y'_{ij}
  :=
  \max\{y_{ij},|Z_{ij}|\}
\]
and
\begin{equation}
  Y'
  :=
  \sum_{i,j=0}^{d-1}
  Z_{ij}|ii\rangle\langle jj|
  +
  \sum_{0\leq i<j\leq d-1}
  y'_{ij}
  \bigl(
    |ij\rangle\langle ij|
    +
    |ji\rangle\langle ji|
  \bigr).
  \label{eq:supp-terminal-witness}
\end{equation}
Since $y'_{ij}\geq y_{ij}$, we have $Y'\geq Y\geq0$, and hence
\[
  -Y'\leq X^\Gamma\leq Y'.
\]
Every $2\times2$ block of $(Y')^\Gamma$ has eigenvalues
$y'_{ij}\pm|Z_{ij}|$, which are nonnegative, and the remaining
one-dimensional blocks are $Z_{ii}\geq0$.  Therefore
\[
  (Y')^\Gamma\geq0.
\]
Finally, by \eqref{eq:supp-Y-Gamma-norm} and the condition $\Vert Y^\Gamma \Vert_1 \leq 1$ from Eq.~\eqref{eq:supp-ppt2},
\[
  \operatorname{Tr}Y'
  =
  \operatorname{Tr}Z
  +
  2\sum_{0\leq i<j\leq d-1}y'_{ij}
  =
  \|Y^\Gamma\|_1
  \leq1.
\]
Thus $Y'$ fulfills the conditions in \eqref{eq:supp-ppt1-tilde}, proving that
$X\in\widetilde{\mathrm{PPT}}_1$.  This establishes
\eqref{eq:supp-reduced-cone-equality}.

\medskip
\noindent\textit{Step 4: conclusion.}
Combining \eqref{eq:supp-reduce-ppt2},
\eqref{eq:supp-commuting}, and
\eqref{eq:supp-reduced-cone-equality}, we obtain
\begin{align}
  \inf_{\tau\in\mathrm{PPT}_2}D_{\mathbb M}(\rho_{\text{CSA}}\Vert \tau)
  &=
  \inf_{\tau\in\mathrm{PPT}_2\cap\mathcal A_d}
  D(\rho_{\text{CSA}}\Vert\tau) \label{eq:step4_1} \\ 
  &=
  \inf_{\tau\in\widetilde{\mathrm{PPT}}_1\cap\mathcal A_d}
  D(\rho_{\text{CSA}}\Vert\tau)\\
  &=
  \inf_{\tau\in\widetilde{\mathrm{PPT}}_1}
  D(\rho_{\text{CSA}}\Vert\tau), \label{eq:step4_3}
\end{align}
where the last equality is \eqref{eq:supp-reduce-ppt1}.  The general
bounds \eqref{eq:supp-general-bounds} place
$E_R^\infty(\rho_{\text{CSA}})$ between the first and last expressions, which are
equal.  This proves \eqref{eq:supp-main-result} and therefore Lemma \ref{lemma:family} for the canonical symmetric-antisymmetric family.
\hfill$\square$

We now turn to the second family, the hyperoctahedral family from Eq.~\eqref{eq:main_HO_family}. Note that this family was already introduced in \cite{nechita2026random}, where also a corresponding twirl $\mathcal{T}_{\text{HO}}$ was defined, which maps states onto this family and $\mathcal{T}_{\text{HO}}(\rho_{\text{HO}}) = \rho_{\text{HO}}$. We can therefore use data processing as before. We denote the range of this twirl as 
\begin{align}
    \mathcal{B}_d = \text{range}(\mathcal{T}_{\text{HO}}).
\end{align}

The only thing left to show is that 
\begin{align}
    \PPT_1 \cap \ \mathcal{B}_d = \widetilde \PPT_1 \cap \mathcal{B}_d.
\end{align}

Thus, we again assume $X \in \PPT_1 \cap \ \mathcal{B}_d$. As $X \in \PPT_1$, $\Vert X^\Gamma \Vert_1 \leq 1$. To see that $X \in \widetilde{\PPT}_1$ choose the witness $Y = |X^\Gamma| \geq 0$. Clearly, $-Y \leq X^\Gamma \leq Y$ and $\Tr(Y) = \Tr(|X^\Gamma|) = \Vert X^\Gamma \Vert \leq 1$. It remains to show that $Y^\Gamma \geq 0$.
We do this, by showing that every $0\leq X\in \mathcal{B}_d$ has the property $\mathcal{N}_2(X) \ge 0$, since then $Y^\Gamma = |X^\Gamma|^\Gamma \geq 0$.

\begin{lemma}
\label{lem:HO_bineg}
Every positive semidefinite operator $X \in \mathcal{B}_d$  satisfies
\begin{equation}
 |X^\Gamma|^\Gamma\geq0.
 \label{eq:HO_bipositive}
\end{equation}
\end{lemma}

\begin{proof}
The following partial transposes are useful
\begin{equation}\label{eq:Gamma}
 \id^\Gamma=\id,\qquad \Delta^\Gamma=\Delta,\qquad
 F^\Gamma= d \ketbra{\Phi_d}{\Phi_d}.
\end{equation}

Recall that the projectors were defined as
\begin{align}
    &\Pi_0 = \ketbra{\Phi_d}{\Phi_d}, \\
    &\Pi_1 = \frac{\id + F}{2} - \Delta,  \\
    &\Pi_2 = \frac{\id - F}{2},\\
    & \Pi_3 = \Delta - \ketbra{\Phi_d}{\Phi_d},
\end{align}
with $\Delta = \sum_i \ketbra{ii}{ii}$.
Write \(X=\sum_kx_k \Pi_k\), with \(x_k\geq0\), and introduce
\begin{equation}
 \begin{aligned}
 u&=\frac{x_0+(d-1)x_3}{d},&
 q&=\frac{x_0-x_3}{d},\\
 v&=\frac{x_1+x_2}{2},&
 p&=\frac{x_1-x_2}{2}.
 \end{aligned}
 \label{eq:coordinates}
\end{equation}
Positivity of \(X\) gives
\begin{equation}
 u,v\geq0,\qquad -\frac{u}{d-1}\leq q\leq u,\qquad |p|\leq v.
 \label{eq:coordinate-bounds}
\end{equation}
Using Eq.~\eqref{eq:Gamma}, the eigenvalues of \(X^\Gamma\), in the same
projector order, are
\begin{equation}
 a_0=u+(d-1)p,\qquad a_1=v+q,\qquad a_2=v-q,\qquad a_3=u-p.
 \label{eq:transpose-eigenvalues}
\end{equation}
Set \(b_k=|a_k|\), so \(|X^\Gamma|=\sum_k b_kP_k\), and define (analogously to Eq.~\eqref{eq:coordinates})
\begin{equation}
 U=\frac{b_0+(d-1)b_3}{d},\quad
 Q=\frac{b_0-b_3}{d},\quad
 V=\frac{b_1+b_2}{2},\quad
 P=\frac{b_1-b_2}{2}.
 \label{eq:absolute-coordinates}
\end{equation}
A second use of Eq.~\eqref{eq:transpose-eigenvalues} shows that the
eigenvalues of \((|X^\Gamma|)^\Gamma\) are
\begin{equation}
 U+(d-1)P,\qquad V+Q,\qquad V-Q,\qquad U-P.
 \label{eq:HO_bipositive-eigenvalues}
\end{equation}
The triangle and reverse triangle inequalities imply
\begin{equation}
 U= \frac{\abs{a_0}+(d-1)\abs{a_3}}{d}\geq
 \left|\frac{a_0+(d-1)a_3}{d}\right|=u,\qquad
 |Q|\leq\frac{|a_0-a_3|}{d}=|p|\leq v.
 \label{eq:UQ-bounds}
\end{equation}
Also, from \(v\geq0\) and elementary inequalities it follows that
\begin{equation}
 V=\max\{v,|q|\},\qquad
 P=\text{sgn}(q)\min\{v,|q|\}.
 \label{eq:VP-identities}
\end{equation}
Thus \(V\geq v\geq|Q|\), proving \(V\pm Q\geq0\).
If \(q\geq0\), then \(0\leq P\leq q\leq u\leq U\).
If \(q\leq0\), then
\[
 0\leq-(d-1)P\leq-(d-1)q\leq u\leq U.
\]
In either case the first and the last eigenvalue in
\eqref{eq:HO_bipositive-eigenvalues} is nonnegative as well.
This proves \eqref{eq:HO_bipositive}.
\end{proof}

Now, since $\PPT_1 \cap \ \mathcal{B}_d = \widetilde \PPT_1 \cap \mathcal{B}_d$ and every $\sigma\in\mathcal B_d$ commutes with any hyperoctahedral state $\rho_{\text{HO}}$ and the twirl preserves the constraints of both relevant sets, the same arguments as in \eqref{eq:step4_1}-\eqref{eq:step4_3} go through and we arrive at Lemma \ref{lemma:family}.

\section{Counterexample for the $\chi$ hierarchy collapse.}
 
Here we construct, for each choice of $m \in \mathbb{N}$, $m \geq 3$, a quantum state $\rho_m$ such that $\chi_{m-2}(\rho_m) < \chi_{m-1}(\rho_m)$, with $\chi_m$ defined in Eq.~\eqref{eq:chi_hierarchy}. This shows that the conjectured equality $\chi_2(\rho) \overset{?}{=} \chi_3(\rho)$ from Ref.~\cite{LamiMeleRegula2025} is false and, more importantly, that the hierarchy does not collapse at any finite level of $m$. Finally, we show that the state $\rho_m$ is $m$-multinegative but not $(m+1)$-multinegative, therefore yielding examples of states with arbitrarily high levels of multinegativity. 

\subsection{The class of product Bell diagonal states}

Denote by $\Phi_0 = \ketbra{\phi^+}{\phi^+}, \Phi_1 = \ketbra{\phi^-}{\phi^-}, \Phi_2 = \ketbra{\psi^+}{\psi^+}, \Phi_3 = \ketbra{\psi^-}{\psi^-}$ the projectors onto the usual two-qubit Bell states. For some $n \geq 1$, we consider the class of states of $n$ qubit pairs, given by
\begin{align}\label{eq:pbds}
    \rho_v = \sum_{z\in\mathbb{Z}_4^n} v_z P_z,
\end{align}
where $v_z \geq 0$, $\sum_z v_z = 1$ and $P_z = \bigotimes_{j=1}^n \Phi_{z_j}$.
Note that
\begin{align}
    \Phi_i^\Gamma = \frac12 \sum_{j=0}^3 (-1)^{\delta_{i,3-j}} \Phi_j.
\end{align}
Thus, partial transposition yields $\rho_v^\Gamma = \rho_w$ with
\begin{align}
    & w_z = \sum_{w\in \mathbb{Z}_4^n} T_{z_1w_1} T_{z_2w_2} \ldots T_{z_nw_n} v_w \\
     \Leftrightarrow \quad & w_{\phantom{z}} = T^{\otimes n}v
\end{align}
with 
\begin{align}\label{eq:Tmap}
    T = \frac12 \begin{pmatrix} 1 & 1 & 1 & -1 \\ 1 & 1 & -1 & 1 \\ 1 & -1 & 1 & 1\\ -1 & 1 & 1 & 1 \end{pmatrix}.
\end{align}
Thus, the family of product Bell diagonal states is closed under partial transposition. Due to the orthogonality of the $P_z$, positivity of $\rho_v$ amounts to $v_z \geq 0$, whereas $\rho_v$ is PPT if $(T^{\otimes n}v)_z \geq 0$ for all $z$.

For a given two-qubit state $\tau$, consider the Pauli twirl, i.e.,
\begin{align}
    \mathcal{T}(\tau) := \frac14\sum_{\alpha=0}^3 (\sigma_\alpha \otimes \sigma_\alpha) \tau (\sigma_\alpha \otimes \sigma_\alpha),
\end{align}
where $\sigma_0 = \id, \sigma_1 = X,  \sigma_2 = Y,  \sigma_3 = Z$ denote the Pauli matrices. The Pauli twirl maps every state to a Bell diagonal one, and 
\begin{align}
    \mathcal{T}^{\otimes n}(\rho_v) = \rho_v
\end{align}
for all product Bell diagonal states. Thus, the twirl commutes with partial transposition and is trace preserving. Finally, because it maps operators to convex combinations of local unitary rotations of itself, it preserves matrix inequalities, i.e., if $\rho_1 \leq \rho_2$, then also $\mathcal{T}^{\otimes n}(\rho_1) \leq \mathcal{T}^{\otimes n}(\rho_2)$. This can be exploited to reduce the optimization for $\chi_p$ in Eq.~\eqref{eq:chi_hierarchy} for any product Bell diagonal input $\rho_v$: For any feasible tuple of matrices $(S_{-1} = \rho_v,S_0, \ldots, S_p)$, we can construct another feasible tuple $(\mathcal{T}^{\otimes n}(S_{-1}) = \rho_v, \mathcal{T}^{\otimes n}(S_0)\ldots, \mathcal{T}^{\otimes n}(S_p))$. Thus, it suffices to optimize over product Bell diagonal matrices $S_i = \sum_{z\in \mathbb{Z}_4^n} s_{i,z} P_z$.

The first sandwich constraint $-S_0 \leq \rho_v^\Gamma \leq S_0$ and the remaining ones, $-S_i \leq S_{i-1}^\Gamma \leq S_i$ then become 
\begin{align}
    -s_{0,z} &\leq (T^{\otimes n}v)_z \phantom{_{i-1}}\leq s_{0,z}\qquad \forall z\in \mathbb{Z}_4^n, \label{eq:pbs_chi1}\\
-s_{i,z} &\leq (T^{\otimes n}s_{i-1})_z \leq s_{i,z}\qquad \forall z\in \mathbb{Z}_4^n, i=1,\ldots,p.\label{eq:pbs_chi2}
\end{align}

To save space, let us write for any coefficient vectors $x$ and $y$ that  $x \geq y$ if $x_z \geq y_z$ for all $z$, and $|x|$ denotes the vector with entries $|x|_z = |x_z|$ for all $z$, i.e., inequalities and absolute values are to be understood component wise. We can then write above Equations~\eqref{eq:pbs_chi1} and \eqref{eq:pbs_chi2} as 
\begin{align}
    s_0 \geq |T^{\otimes n} v|\quad\text{and} \quad s_i \geq |T^{\otimes n}s_{i-1}| \text{ for }i=1,\ldots,p.
\end{align}
Furthermore, define the all-ones vector $\boldsymbol{1}^{(n)}$ (the superscript making explicit the number of Bell pairs we consider) with entries
\begin{align}
    (\boldsymbol{1}^{(n)})_z = 1\quad \forall z\in \mathbb{Z}_4^n.
\end{align}
With this, we can write $\sum_z v_z = \boldsymbol{1}^{(n)} \cdot v$.

With these abbreviations, we can write the optimization for $\chi$ as the linear program
\begin{align}\label{eq:chi_pbd}
    \chi_p(\rho_v) = \min\left\{ \boldsymbol{1}^{(n)} \cdot s_p \,:\, s_i \geq |T^{\otimes n} s_{i-1}| \text{ for } i = 0,\ldots,p \text{ and } s_{-1} = v\right\}.
\end{align}

Finally, we define on the level of the product Bell diagonal coefficients the map
\begin{align}
    F^{(n)}(x):=|T^{\otimes n}x|,
\end{align}
and its $k$-fold application as
\begin{align}
    F^{(n)}_k(x):=\underbrace{F^{(n)} \circ F^{(n)} \circ \ldots \circ  F^{(n)}}_{k \; \text{times}}(x).
\end{align}

Thus, $F^{(n)}$ is reminiscent to the absolute partial transposition map $\mathcal{N}$ in Eq.~\eqref{eq:Nmap}, but $F^{(n)}$ applies first the partial transposition, followed by taking the absolute value. Thus,
\begin{align}\label{eq:NkFk}
    \mathcal{N}_k(\rho_v) = \rho_{T^{\otimes n}F^{(n)}_{k-1}(v)}.
\end{align}

\subsection{Construction of a $\chi$ separating state.}

Fix $m \geq 1$.
Start with a single qubit pair and consider the non-normalized, positive semidefinite operator $\rho_{v_1}$ with 
\begin{align}
    v_1 = \begin{pmatrix}2 \\ 0 \\ 0 \\ 0\end{pmatrix},
\end{align}
i.e., $\rho_{v_1} = 2\ketbra{\phi^+}{\phi^+}$. From this, we construct an $m$ qubit pair, product Bell diagonal state by setting recursively
\begin{align}
    \rho_{v_k} := (|\rho_{v_{k-1}}^\Gamma|+\rho_{v_{k-1}}^\Gamma) \otimes \Phi_0+ (|\rho_{v_{k-1}}^\Gamma|-\rho_{v_{k-1}}^\Gamma)\otimes \Phi_1,
\end{align}
with, again, $\Phi_0 = \ketbra{\phi^+}{\phi^+}$, $\Phi_1 = \ketbra{\phi^-}{\phi^-}$.
As in each step we only add a Bell pair, the result is again product Bell diagonal. On the level of the coefficients, we can equivalently write
\begin{align}\label{eq:vm}
    v_k := |\tilde{v}_{k-1}| \otimes \begin{pmatrix}1 \\ 1 \\ 0 \\ 0\end{pmatrix} + \tilde{v}_{k-1} \otimes \begin{pmatrix}1 \\ -1 \\ 0 \\ 0\end{pmatrix},
\end{align}
where $\tilde{v}_k := T^{\otimes k}v_k$. 
Note that, due to $v_1 \geq 0$ and $|x| - x \geq 0$, the $4^m$ entries of $v_m$ are non-negative and therefore $\rho_{v_m} \geq 0$ (but not normalized).

Let us establish an important property of $\rho_{v_m}$ on which the subsequent results rely:
\begin{lemma}\label{lem:Fm}
For every $m\geq 1$, the vector $v_m$ reaches the flat distribution $\boldsymbol{1}^{(m)}$ after exactly $m$ applications of $F^{(m)}$:
\begin{align}
  F^{(m)}_k(v_m)&=\boldsymbol{1}^{(m)}\quad(k\geq m),  \\
  F^{(m)}_k(v_m)&\neq\boldsymbol{1}^{(m)}
  \quad(0\leq k<m).
\end{align}
\end{lemma}
\begin{proof}
    For $m=1$, the claim is obviously true, since 
    \begin{align}
        F^{(1)}_1(v_1) = \left|T\begin{pmatrix}2 \\ 0 \\ 0 \\ 0\end{pmatrix}\right| = \left|\begin{pmatrix}1 \\ 1 \\ 1 \\ -1\end{pmatrix}\right| = \boldsymbol{1}^{(1)}
    \end{align}
    (and likewise, since $T^{\otimes m} \boldsymbol{1}^{(m)} = \boldsymbol{1}^{(m)}$, the same is true for more applications of $F^{(1)}$),
    whereas
    \begin{align}
    F^{(1)}_0(v_1) = v_1 =\begin{pmatrix}2 \\ 0 \\ 0 \\ 0\end{pmatrix}\neq \boldsymbol{1}^{(1)}.
    \end{align}
    We prove the case of $m>1$ by induction. To that end, note that
    \begin{align}
        T\begin{pmatrix}1 \\ 1 \\ 0 \\ 0\end{pmatrix} = \begin{pmatrix}1 \\ 1 \\ 0 \\ 0\end{pmatrix},\qquad T\begin{pmatrix}1 \\ -1 \\ 0 \\ 0\end{pmatrix} = \begin{pmatrix}0 \\ 0 \\ 1 \\ -1\end{pmatrix},\\
        T\begin{pmatrix}0 \\ 0 \\ 1 \\ 1\end{pmatrix} = \begin{pmatrix}0 \\ 0 \\ 1 \\ 1\end{pmatrix},\qquad T\begin{pmatrix}0 \\ 0 \\ 1 \\ -1\end{pmatrix} = \begin{pmatrix}1 \\ -1 \\ 0 \\ 0\end{pmatrix}.
    \end{align}
Thus, 
\begin{align}
    T^{\otimes m} v_m&=T^{\otimes (m-1)}|\tilde{v}_{m-1}| \otimes T\begin{pmatrix}1 \\ 1 \\ 0 \\ 0\end{pmatrix} + T^{\otimes(m-1)}\tilde{v}_{m-1} \otimes T\begin{pmatrix}1 \\ -1 \\ 0 \\ 0\end{pmatrix} \\
    &= T^{\otimes (m-1)}|\tilde{v}_{m-1}| \otimes \begin{pmatrix}1 \\ 1 \\ 0 \\ 0\end{pmatrix} + v_{m-1} \otimes \begin{pmatrix}0 \\ 0 \\ 1 \\ -1\end{pmatrix}.
\end{align}
The two terms on the right-hand side are mutually orthogonal. Thus, taking component wise absolute value, yields
\begin{align}
    F^{(m)}(v_m)=|T^{\otimes m} v_m|&= |T^{\otimes (m-1)}|\tilde{v}_{m-1}|| \otimes \begin{pmatrix}1 \\ 1 \\ 0 \\ 0\end{pmatrix} + |v_{m-1}| \otimes \begin{pmatrix}0 \\ 0 \\ 1 \\ 1\end{pmatrix} \\
    &= F^{(m-1)}_2(v_{m-1})\otimes \begin{pmatrix}1 \\ 1 \\ 0 \\ 0\end{pmatrix} + v_{m-1}\otimes \begin{pmatrix}0 \\ 0 \\ 1 \\ 1\end{pmatrix}.
\end{align}
Here, we have used $T^2 = \id$ and that the entries of $v_{m-1}$ are non-negative. Applying the map $F^{(m)}$ again to both sides of the equation and making use of $T(0,0,1,1)^{\text{T}} = (0,0,1,1)^{\text{T}}$ shows that we only keep applying $F^{(m-1)}$, which yields for any $k \geq 1$,
\begin{align}\label{eq:Fk_recursive}
    F^{(m)}_k(v_m)
    &= F^{(m-1)}_{k+1}(v_{m-1})\otimes \begin{pmatrix}1 \\ 1 \\ 0 \\ 0\end{pmatrix} + F^{(m-1)}_{k-1}(v_{m-1})\otimes \begin{pmatrix}0 \\ 0 \\ 1 \\ 1\end{pmatrix}.
\end{align}
First, set $k = m$. By assumption of induction, $F^{(m-1)}_{m+1}(v_{m-1}) = \boldsymbol{1}^{(m-1)}$ and $F^{(m-1)}_{m-1}(v_{m-1}) = \boldsymbol{1}^{(m-1)}$. Thus, Eq.~\eqref{eq:Fk_recursive} yields
\begin{align}
    F^{(m)}_m(v_m) &= \boldsymbol{1}^{(m-1)} \otimes \begin{pmatrix}1 \\ 1 \\ 0 \\ 0\end{pmatrix} + \boldsymbol{1}^{(m-1)} \otimes \begin{pmatrix}0 \\ 0 \\ 1 \\ 1\end{pmatrix} = \boldsymbol{1}^{(m)}.
\end{align}
On the other hand, for $k = m-1$, $F^{(m-1)}_{m-2} \neq \boldsymbol{1}^{(m-1)}$, thus
\begin{align}
    F^{(m)}_{m-1}(v_m) &= \boldsymbol{1}^{(m-1)} \otimes \begin{pmatrix}1 \\ 1 \\ 0 \\ 0\end{pmatrix} + F^{(m-1)}_{m-2} \otimes \begin{pmatrix}0 \\ 0 \\ 1 \\ 1\end{pmatrix} \neq \boldsymbol{1}^{(m)}.
\end{align}
\end{proof}

Let us continue by normalizing the constructed operators $\rho_{v_m}$. We define
\begin{align}\label{eq:rhom}
    \rho_m := \frac{\rho_{v_m}}{Z_m}
\end{align}
with \begin{align}
    Z_m = \Tr(\rho_{v_m}) = \sum_z v_{m,z} = \boldsymbol{1}^{(m)} \cdot v_m
\end{align}
and $v_m$ constructed from Eq.~\eqref{eq:vm}.
Since $v_m$ is non-zero and non-negative, $Z_m>0$ and $\rho_m$ is a quantum state of $m$ qubit pairs.

Now we are in position to calculate the value of $\chi_{m-1}(\rho_m)$:

\begin{lemma}
For every $m\geq 1$,
\begin{equation}\label{eq:chi_mm1}
  \chi_{m-1}(\rho_m)=\frac{4^m}{Z_m}.
\end{equation}
\end{lemma}

\begin{proof}
Put
\begin{align}
  x_k:=F^{(m)}_k(v_m),
  \qquad (k=0,\ldots,m).
\end{align}
Lemma~\ref{lem:Fm} gives $x_m=\boldsymbol{1}^{(m)}$. Set
\begin{align} \label{eq:si_feasib}
  s_i&:=\frac{x_{i+1}}{Z_m}
  \qquad (i=0,\ldots,m-1),\\
  s_{-1} &:= \frac{v_m}{Z_m}.
\end{align}
Thus, $s_i = x_{i+1}/Z_m = F^{(m)}_{i+1}(v_m)/Z_m = |T^{\otimes m} F^{(m)}_{i}(v_m)|/Z_m = |T^{\otimes m} x_i|/Z_m = |T^{\otimes m}s_{i-1}|$. This means that the $s_i$ constitute a feasible point in the optimization of $\chi_{m-1}$ in Eq.~\eqref{eq:chi_pbd}. Thus,
\begin{align}
    \chi_{m-1}(\rho_m) \leq \boldsymbol{1}^{(m)} \cdot s_{m-1} = \boldsymbol{1}^{(m)} \cdot \frac{x_m}{Z_m} = \frac{\boldsymbol{1}^{(m)} \cdot \boldsymbol{1}^{(m)}}{Z_m} = \frac{4^m}{Z_m}.
\end{align}

It remains to show that no other feasible point can reach a lower objective value.  For $i=0,\ldots,m-1$, let
\begin{equation}
  D_i:=\operatorname{diag}(d_{i,z})_z,
\end{equation}
where $d_{i,z}=1$ if $(T^{\otimes m}x_i)_z\geq 0$ and $d_{i,z}=-1$ otherwise. Thus $D_i^2=I$ and
\begin{equation}
  x_{i+1}=|T^{\otimes m}x_i|=D_i T^{\otimes m} x_i,
\end{equation}
and, since $D_i^2 = \id$, 
\begin{equation}\label{eq:xixip1}
     T^{\otimes m} D_i x_{i+1}=x_i.
\end{equation}
Consider any feasible point $s_{-1},s_0,\ldots,s_{m-1}$ with $s_{-1}=v_m/Z_m$. Componentwise,
\begin{equation}\label{eq:sisim1}
  s_i\geq| T^{\otimes m}s_{i-1}|\geq D_iT^{\otimes m}s_{i-1},
\end{equation}
where the second inequality stems from the fact that the sign matrix $D_i$ might make some entries of $T^{\otimes m}s_{i-1}$ negative.

Since $x_{i+1}\geq 0$, Eqs.~\eqref{eq:xixip1} and \eqref{eq:sisim1} imply
\begin{align}
  x_{i+1}\cdot s_i = x_{i+1}^{\text{T}} s_i
  &\geq x_{i+1}^{\text{T}} D_iT^{\otimes m}s_{i-1},
  \\
  &=(T^{\otimes m}D_ix_{i+1})^{\text{T}} s_{i-1}=x_i^{\text{T}} s_{i-1} = x_i\cdot s_{i-1}.
\end{align}
Starting with $i=m-1$ gives
\begin{equation}\label{eq:chimm1_lb}
  \boldsymbol{1}^{(m)} \cdot s_{m-1}
  =x_m \cdot s_{m-1}
  \geq \ldots \geq x_0 \cdot s_{-1}
  =v_m\cdot \frac{v_m}{Z_m} = \frac{\lVert v_m\rVert_2^2}{Z_m}.
\end{equation}
The $2$-norm of any $x$ is invariant under the application of $F^{^{(m)}}$, since 
\begin{align}\label{eq:Fknorm}
    \lVert F^{(m)}(x) \rVert_2^2 = \sum_z |T^{\otimes m} x|_z^2,
\end{align} and partial transposition keeps the $2$-norm invariant. Thus,
\begin{equation}
  \lVert v_m\rVert_2^2=\lVert F^{(m)}_m(v_m)\rVert_2^2
  =\lVert x_m\rVert_2^2
  =\lVert\boldsymbol{1}^{(m)}\rVert_2^2
  =4^m.
\end{equation}
Thus, choosing the optimal $s_i$, we get from Eq.~\eqref{eq:chimm1_lb}
\begin{align}
    \chi_{m-1}(\rho_m) = \boldsymbol{1}^{(m)} \cdot s_{m-1} \geq \frac{\lVert v_m\rVert_2^2}{Z_m} = \frac{4^m}{Z_m},
\end{align}
matching the upper bound from before and proving the claim.
\end{proof}

Finally, we show that level $\chi_{m-2}$ yields a strictly smaller value.

\begin{theorem}[$\chi$-gap at every finite level]
For every integer $m\geq 3$, the state $\rho_m$ from Eq.~\eqref{eq:rhom} fulfills
\begin{equation}
  \chi_{m-2}(\rho_m)<\chi_{m-1}(\rho_m).
\end{equation}
\end{theorem}

\begin{proof}
Fix $m\geq 3$ and consider again the feasible point $s_{-1},s_0,\ldots,s_{m-1}$ in Eq.~\eqref{eq:si_feasib} for the optimization of $\chi_{m-1}$. Discarding $s_{m-1}$ yields a feasible point for $\chi_{m-2}$, showing 
\begin{equation}\label{eq:chi_mm2_bound}
  \chi_{m-2}(\rho_m)
  \leq \boldsymbol{1}^{(m)}\cdot s_{m-2} =\frac{\boldsymbol{1}^{(m)} \cdot x_{m-1}}{Z_m}.
\end{equation}
By Lemma~\ref{lem:Fm}, $x_{m-1}\neq\boldsymbol{1}_m$. On the other hand, as shown in Eq.~\eqref{eq:Fknorm}, $F^{(m)}$ preserves the $2$-norm and $F^{(m)}(x_{m-1})=\boldsymbol{1}^{(m)}$, so
\begin{align}
  \lVert x_{m-1}\rVert_2
  =\lVert\boldsymbol{1}^{(m)}\rVert_2
  =\sqrt{4^m}.
\end{align}
Because $x_{m-1}\geq 0$, the Cauchy-Schwarz inequality gives
\begin{align}
  \boldsymbol{1}^{(m)} \cdot x_{m-1}
  \leq\lVert\boldsymbol{1}^{(m)}\rVert_2\lVert x_{m-1}\rVert_2
  =4^m.
\end{align}
Equality would force $x_{m-1}$ to be proportional to $\boldsymbol{1}_m$, which is a contradiction. Therefore the inequality is strict. Combining this with Eqs.~\eqref{eq:chi_mm1} and \eqref{eq:chi_mm2_bound} yields
\begin{align}
  \chi_{m-2}(\rho_m)
  <\frac{4^m}{Z_m}
  =\chi_{m-1}(\rho_m).
\end{align}
\end{proof}

Some comments are in order. 

First, at level $m-1$, the hierarchy collapses for $\rho_m$. This can be seen from the fact that the optimal feasible point at level $\chi_{m-1}$ fulfills $s_{m-1} = x_m/Z_m = \boldsymbol{1}^{(m)}/Z_m$. Therefore, in the optimization of $\chi_m$, one can simply choose $s_m = s_{m-1}$ (and the other $s_i$ as before) as a feasible point, showing $\chi_m(\rho_m) \leq 4^m/Z_m$, implying collapse. 

Second, the normalization constant $Z_m$ can be calculated explicitly to read $Z_m = 2^m \binom{m}{\lfloor m/2\rfloor}\sim 4^m /\sqrt{\pi m/2}$. Thus,
\begin{align}
    \chi_{m-1}(\rho_m) = \chi_m(\rho_m) = 2^{E_{c,\text{PPT}}^{\text{exact}}(\rho_m)} = \frac{2^m}{\binom{m}{\lfloor m/2 \rfloor}} \sim \sqrt{\pi m/2}.
\end{align}
Furthermore, the optimal point for the optimization of $\chi_{m-1}(\rho_m)$ can be truncated to a feasible point of $\chi_{m-2}(\rho_m)$. Evaluating its value yields
\begin{align}
    \chi_{m-2}(\rho_m) \leq \frac{2^m-1}{\binom{m}{\lfloor m/2 \rfloor}},
\end{align}
yielding a gap of at least
\begin{align}
    \chi_{m-1}(\rho_m) - \chi_{m-2}(\rho_m) \geq \frac{1}{\binom{m}{\lfloor m/2 \rfloor}}.
\end{align}

Third, the state $\rho_m$ is $m$-multinegative, but not $(m+1)$-multinegative. This can be seen by recalling from Eq.~\eqref{eq:NkFk} that $\mathcal{N}_m(Z_m\rho_m) = \rho_{T^{\otimes m}F^{(m)}_{m-1}(v_m)}$. Using Eq.~\eqref{eq:Fk_recursive} with $k=m-1$, we find
\begin{align}
    T^{\otimes m}F^{(m)}_{m-1}(v_m)
    &= T^{\otimes (m-1)}F^{(m-1)}_{m}(v_{m-1})\otimes \begin{pmatrix}1 \\ 1 \\ 0 \\ 0\end{pmatrix} + T^{\otimes (m-1)}F^{(m-1)}_{m-2}(v_{m-1})\otimes \begin{pmatrix}0 \\ 0 \\ 1 \\ 1\end{pmatrix} \\
    &= \boldsymbol{1}^{(m-1)}\otimes \begin{pmatrix}1 \\ 1 \\ 0 \\ 0\end{pmatrix} + T^{\otimes (m-1)}F^{(m-1)}_{m-2}(v_{m-1})\otimes \begin{pmatrix}0 \\ 0 \\ 1 \\ 1\end{pmatrix}.
\end{align}
This has non-negative coefficients, if and only if $T^{\otimes (m-1)}F^{(m-1)}_{m-2}(v_{m-1}) \geq 0$. Repeating this calculation recursively, yields
\begin{align}
    & T^{\otimes m}F^{(m)}_{m-1}(v_m) \geq 0 \\
    \Leftrightarrow\quad & T^{\otimes (m-1)}F^{(m-1)}_{m-2}(v_{m-1}) \geq 0 \\
    \Leftrightarrow\quad &\ldots \\
    \Leftrightarrow\quad & TF^{(1)}_{0}(v_{1}) \geq 0 \\
    \Leftrightarrow\quad & Tv_{1} = (1,1,1,-1){^{\text{T}}} \geq 0,
\end{align}
which clearly is not positive. Thus, $\mathcal{N}_m(\rho_m) \ngeq 0$. However, Eq.~\eqref{eq:Fk_recursive} with $k=m$ shows that $\mathcal{N}_{m+1}(\rho_m) \propto \id \geq 0$.

Fourth and finally, the local dimension of the state $\rho_m$ is given by $2^m$, i.e., it scales exponentially with the negativity-level. Numerically, however, we were able to find $m$-multinegative states up to $m=7$ with local dimension $\lceil (m+3)/2\rceil$ within the class of Grid states \cite{lockhart2021combinatorial, lockhart2018entanglement}, suggesting that increasing the level of $m$-multinegativity requires only a linear increase in local dimension. 

\end{document}